%% file: main.tex
\documentclass[11pt,english]{article}
\usepackage[margin=0.75in]{geometry}
\usepackage{hyperref}
\hypersetup{
     colorlinks=true,
     linkcolor=blue,
     citecolor=blue,
}
\usepackage[ruled,vlined,linesnumbered]{algorithm2e}

\usepackage{pst-node}
\usepackage{pst-coil}
\usepackage{auto-pst-pdf}
 \usepackage[english]{babel}
\usepackage{bm,xfrac,amsmath,amsthm,soul, comment, xspace, mathtools, amsfonts, amssymb}
\usepackage{footnote}
\usepackage{nicefrac,bm}
\usepackage{thmtools}
\usepackage{thm-restate}
\usepackage{graphicx}
\usepackage{pifont,multirow}
\usepackage{textcomp}
\usepackage[switch]{lineno}
\newcommand{\truncate}[2]{#1_{\downarrow #2}}
\title{$\nicefrac{1}{2}$-Approximate $\MMS$ Allocation for Separable Piecewise Linear Concave Valuations }
\date{}
\author{
Chandra Chekuri\thanks{University of Illinois at Urbana-Champaign}\\
\texttt{chekuri@illinois.edu}\and
Pooja Kulkarni\thanks{University of Illinois at Urbana-Champaign.}\\ \texttt{poojark2@illinois.edu} \and
Rucha Kulkarni\thanks{University of Illinois at Urbana-Champaign.}\\\texttt{ruchark2@illinois.edu} \and
Ruta Mehta\thanks{University of Illinois at Urbana-Champaign}\\ \texttt{rutameht@illinois.edu}
}

\let\origappendix\appendix 
\renewcommand\appendix{\clearpage\pagenumbering{arabic}\origappendix}
\input{author_macros}
\begin{document}
\maketitle
\begin{abstract}
We study fair distribution of a collection of $m$ indivisible goods among a group of $n$ agents, using the widely recognized fairness principles of Maximin Share $(\MMS)$ and Any Price Share $(\APS)$. These principles have undergone thorough investigation within the context of additive valuations. We explore these notions for valuations that extend beyond additivity.

First, we study approximate $\MMS$ under the separable (piecewise-linear) concave ($\SPLC$) valuations, an important class generalizing additive, where the best known factor was $1/3$-$\MMS$. We show that $\sfrac{1}{2}$-$\MMS$ allocation exists and can be computed in polynomial time, significantly improving the state-of-the-art.
We note that $\SPLC$ valuations introduce an elevated level of intricacy in contrast to additive. For instance, the $\MMS$ value of an agent can be as high as her value for the entire set of items.\footnote{Also the equilibrium computation problem, which is polynomial-time for additive valuations, becomes intractable for $\SPLC$ \cite{SPLC1}.} We use a relax-and-round paradigm that goes through competitive equilibrium and LP relaxation. Our result extends to give (symmetric) $\sfrac{1}{2}$-$\APS$, a stronger guarantee than $\MMS$. 

$\APS$ is a stronger notion that generalizes $\MMS$ by allowing agents with arbitrary entitlements. We study the approximation of $\APS$ under submodular valuation functions. We design and analyze a simple greedy algorithm using concave extensions of submodular functions. We prove that the algorithm gives a $\sfrac{1}{3}$-$\APS$ allocation which matches the best-known factor by \cite{uziahu2023fair}. Concave extensions are hard to compute in polynomial time and are, therefore, generally not used in approximation algorithms. Our approach shows a way to utilize it within analysis (while bypassing its computation), and hence might be of independent interest. 
\end{abstract}
\paragraph{Acknowledgements.} Chandra Chekuri is partially supported by NSF grant CCF-1910149. Pooja Kulkarni, Rucha Kulkarni and Ruta Mehta are supported by NSF CAREER Award CCF 1750436.
\input{1.intro}
\input{2.prelims}
\input{3.upperbounds}
\input{4.splc}
\input{5.submodular}

\bibliographystyle{alpha}{}
\bibliography{literature}
\input{6.appendix}

\end{document}

%% file: author_macros.tex
\theoremstyle{plain}
\newtheorem{remark}{Remark}[section]
\newtheorem{example}{Example}[section]
\newtheorem{lemma}{Lemma}[section]
\newtheorem*{lemma*}{Lemma}

\newtheorem{theorem}{Theorem}[section]

\newtheorem{claim}{Claim}[section]
\newtheorem*{claim*}{Claim}

\newtheorem{definition}{Definition}[section]

\newcommand{\A}{\mathcal{A}}
\newcommand{\N}{{\mathcal N}}

\newcommand{\Set}{{\mathcal S}}

\newcommand{\M}{{\mathcal M}}
\newcommand{\Vals}{(v_i)_{i\in\N}}

\newcommand{\B}{(b_i)_{i\in\N}}

\newcommand{\MMSins}{(\N,\M,\Vals,\B)}

\newcommand{\MMSinssym}{(\N,\M,\Vals)}
\newcommand{\SPLCins}{(\N, t, (k_j)_{j \in [t]}, \Vals)}

\newcommand{\vecx}{\mathbf{x}}
\newcommand{\vecp}{\mathbf{p}}

\newcommand{\classNP}{{\sf NP}}

\newcommand{\classPPAD}{{\sf PPAD}}
\newcommand{\SPLC}{{\sf SPLC}}

\newcommand{\MMS}{{\sf{MMS}}}

\newcommand{\APS}{{\sf{APS}}}

\newcommand{\PA}{{{A^\pi}}}

\newcommand{\G}{\mathcal{G}}

\let\oldnl\nl
\newcommand{\nonl}{\renewcommand{\nl}{\let\nl\oldnl}}

\DeclareMathOperator*{\argmax}{\arg\!\max}

\renewcommand{\nonl}{\renewcommand{\nl}{\let\nl\oldnl}}
\long\def\symbolfootnote[#1]#2{\begingroup%
\def\thefootnote{\fnsymbol{footnote}}\footnote[#1]{#2}\endgroup}

\newcommand{\Price}{\mathcal{P}}
\newcommand{\allocvalue}{\rho}
\newcommand{\marginalbid}{\mu}

%% file: 1.intro.tex
\section{Introduction}\label{sec:intro}
We consider the problem of fairly allocating a set $\M$ of $m$ indivisible goods among a set $\N$ of $n$ agents with heterogeneous preferences under the popular fairness notions of {\em Maximin share $(\MMS)$} \cite{budish2011combinatorial} and {\em Any Price Share $(\APS)$} \cite{BabaioffEF21}. These notions have been extensively studied for the setting where the agents have additive valuations \cite{barman2017approximation, ghodsi2018fair, garg2018approximating, garg2019improved, AkramiGS23, AkramiG23}. This paper studies the problem beyond additive valuations, particularly for the classical separable-concave valuations \cite{SPLC1, SPLC4} and submodular valuations \cite{ghodsi2018fair, barman2020approximation, uziahu2023fair}. 

   $\MMS$ and $\APS$ are {\em share based} fairness notions, where each agent is entitled to a bundle worth her {\em fair share}. Under $\MMS$, this {\em fair share} of an agent is defined as the maximum value she can guarantee herself under the classical {\em cut-and-choose} mechanism when she is the cutter; she partitions the item set into $n$ bundles and gets to pick last. Therefore, she partitions so that the value of the minimum valued bundle is maximized. Let $\Pi_{\N}(\M)$ denote the set of all allocations of $\M$ among the $n$ agents. If $(A_1,\dots,A_n)$ denotes any allocation into $n$ bundles, and $v_i:2^{\M} \rightarrow \mathbb R_+$ denotes agent $i$'s valuation function, 
    then the $\MMS$ value of agent $i$ is defined as, 
    \[
    \MMS_i = \max_{(A_1,\dots,A_n)\in \Pi_{\N}(\M)} \min_{j\in [n]} v_i(A_j)
    \]
    
    An $\MMS$ allocation is one where every agent $i$ gets a bundle worth at least $\MMS_i$. $\MMS$ treats all agents equally. In some settings it is necessary to consider \emph{weighted} agents, where the weight or entitlement of agent $i$ is $b_i>0$; the weights are normalized to satisfy $\sum_i b_i=1$. It is not straight forward to define a weighted generalization of $\MMS$. To address this, \cite{BabaioffEF21} introduced the notion of $\APS$. This fair share value of agent $i$ is defined as the value she can ensure herself with a budget of $b_i$ when the {\em prices} of the items are chosen adversarially, subject to a normalization constraint that the total sum of prices is $1$. More formally, let $\Price=\{(p_1\dots,p_m)\mid \sum_{j=1}^m p_j=1,\ p_j\ge 0\ \forall j\}$ denote the simplex of price vectors for the $m$ goods. If the budget of agent $i$ is $b_i$, then,   
    \[
    \APS_i = \min_{(p_1,\dots,p_m)\in \Price} \max_{S\subseteq \M: \sum_{j\in S} p_j \le b_i} v_i(S)
    \]

    We note that when $b_i = \frac{1}{n}$, $\APS_i \ge \MMS_i$ \cite{BabaioffEF21}. Thus, allocations that give guarantees with respect to $\APS_i$ at $b_i = \frac{1}{n}$ automatically provide same guarantees for $\MMS_i$.

    Allocations achieving $\MMS$ and $\APS$ shares may not exist even under additive valuations \cite{ProcacciaW14, FeigeST21}. Therefore, the focus has been on finding approximate solutions, where in an $\alpha$-$\MMS$ $(\APS)$ allocation, every agent receives a bundle worth at least $\alpha$ times their $\MMS$ $(\APS)$ value. This problem has been studied extensively for additive valuations (see \cite{AmanatidisABFLMVW22survey} for a survey and pointers) with much progress \cite{ProcacciaW14,garg2018approximating, KulkarniMT21}. It is known that a $(\frac34+\frac{3}{3836})$-$\MMS$ always exists and can be computed in polynomial time \cite{AkramiG23}, while there are examples showing that $\frac{39}{40}$-$\MMS$ may not exist \cite{FeigeST21} even in the setting of three agents.    

Additive valuations are inapplicable if agents have {\em decreasing marginal gains}, a crucial property in practice. This raises the need to go beyond additive. Some well-known classes of valuation functions such as subadditive, fractionally sub-additive i.e. XOS, submodular, and their interesting special cases have been studied in the literature \cite{barman2020existence, LiV21, viswanathan2022yankee}. Here we consider separable (piecewise-linear) concave 
(\SPLC) \cite{SPLC1,SPLC2,SPLC3,SPLC4} and submodular valuations. $\SPLC$ valuations generalize additive valuations, and form a subclass of submodular valuations. Such functions are separable across different types of goods, and concave within each type capturing decreasing marginal gains. Submodular functions allow decreasing marginal gain across all goods. Formally, a real-valued function
$f:2^\M \rightarrow \mathbb{R}$ is submodular iff $f(A \cup \{e\}) - f(A) \ge f(B \cup \{e\}) - f(B)$ for all $A \subset B$ and $e \not \in B$. For submodular valuations, \cite{ghodsi2018fair} gave an algorithm to find a $\frac13$-$\MMS$ allocation based on a certain local search procedure, and \cite{barman2017approximation} showed that a simple round-robin procedure can achieve a $\frac13(1-1/e)$-$\MMS$ allocation. This was recently improved to $\frac{10}{27}$-$\MMS$ by \cite{uziahu2023fair}, who also gave $\frac{1}{3}$-$\APS$ algorithm. These results also apply to separable-concave functions and remain the best known.

In terms of lower bounds, for submodular valuations, even for special cases like assignment valuations and weighted matroid rank valuations it is known that better than $\frac23$-$\MMS$ allocations may not exist \cite{barman2020approximation, kulkarni2023maximin}. For a more general class, namely fractionally subadditive valuations, it is known that better than $\frac12$-$\MMS$ allocations may not exist \cite{ghodsi2018fair}. Closing the gaps for these rich classes of valuations is of much interest. A natural question here is whether $\frac12$-$\MMS$ allocations exist for submodular valuations and it is open even for $\SPLC$ valuations.

One of the difficulties in going beyond additive valuations is the lack of good upper bound on the $\MMS$ ($\APS$) values of the agents. For example, if $v_i$ is additive and $b_i = \frac{1}{n}$ then $\MMS_i \le \APS_i \le  \frac{v_i(\M)}{n}$, while if $v_i$ is separable-concave or submodular then we can have $\MMS_i = v_i(\M)$ (see Example \ref{eg:splc-mms-high} in  Appendix \ref{app:intro}). 

\subsection{Our Results}
In this paper, we develop novel ways to upper bound the $\MMS$ and $\APS$ values for monotone valuation functions via {\em market equilibrium} and {\em concave extensions}. This is our conceptual contribution. We leverage this to obtain two results.

First, we design an LP-relaxation-based method to find a $1/2$-$\MMS$ allocation for $\SPLC$ valuations in polynomial time. Our result is in fact stronger; the algorithm outputs an allocation that gives each agent a value at least $1/2$-$\APS_i$ (for symmetric agents). $\SPLC$ valuations are a special case of submodular valuations, and for the latter, the best known result for $\MMS$ allocations is a very recent result \cite{uziahu2023fair} that yields a $\frac{10}{27}$-$\MMS$ allocation. Thus, for $\SPLC$ valuations we obtain an improved approximation.

Second, we show that a simple greedy algorithm achieves $\frac13$-$\APS$ allocation for submodular valuations. This result was independently obtained in a recent work by \cite{uziahu2023fair}. Their proof technique is conceptually different and uses the bidding game methodology of \cite{BabaioffEF21}. We make use of an argument via the concave extension approach. We believe that this may be of independent interest and may offer helpful insights.

In Appendix \ref{app:intro}, we also point out an example (Example \ref{eg:greedy-half-splc-counter}) to show that a greedy-like approach cannot yield a  $1/2$-MMS allocation for $\SPLC$ valuations. This points to the importance of the LP-based approach that we utilize.

%% file: 2.prelims.tex
\section{Preliminaries}\label{sec:prelims}

We use $[k]$ to denote the set $\{1,2,\cdots,k-1,k\}$.
\subsection{Fair division model} 
We study the problem of fairly dividing a set of $m$ indivisible goods $\M = [m]$ among a set of $n$ agents $\N = [n]$ who can have asymmetric entitlements. The entitlement or weight of agent $i$ is denoted by $b_i,$ and the weights are normalized so that their sum is $1,$ that is, $\sum_{i \in \N} b_i=1.$
The preferences of an agent $i \in \N$ are defined by a valuation function 
$v_i: 2^{\M} \rightarrow \mathbb{R}_{\ge 0}$ over the set of goods. We represent a fair division problem instance by $\MMSins$. When all agents have the same weight, we denote the instance by $\MMSinssym$ and call it the symmetric fair division instance. An allocation $A \coloneqq (A_1, \ldots, A_n)$ is a partition of all the goods among the $n$ agents, i.e. for all $i, j \in [n]$ with $i \neq j$, $A_i \cap A_j = \emptyset$ and $\cup_{i \in [n]}A_i = [m]$. Note that empty parts are allowed. We denote the set of all allocations by $\Pi_{[n]}([m])$. We will also sometimes use fractional allocations of goods. We denote these allocations by $\vecx$ and the allocation of a particular agent, $i \in \N$ by $\vecx_i$. Formally, $\vecx = (x_{ij})_{i \in \N, j \in \M}$ such that $\sum_{i \in \N} x_{ij} \leq 1$ for all $j \in \M$ and $\vecx_i = (x_{ij})_{j \in \M}$ where $(x_{ij})_{i \in \N, j \in \M}$ is a fractional allocation.

Throughout this paper, we assume that the valuation functions are monotone and non-negative. Further, Section \ref{sec:splc} assumes that valuations functions are $\SPLC$ and Section \ref{sec:sub-aps} assumes that the valuation functions are submodular. We define these classes of functions next.

\paragraph{$\SPLC$ valuations.}
Separable-concave, a.k.a. separable piecewise-linear concave ($\SPLC$), valuations is a well-studied class that subsumes additive valuations. Under $\SPLC$ valuations, we have $t$ types of goods, with each good $j \in [t]$ having $k_j$ copies. For an $\SPLC$ valuation function $f(\cdot)$, a value $f_{jk}$ is associated with $k^{th}$ copy of good $j$. The functions are concave so that for all $j \in [t]$, we have $f_{j1} \geq f_{j2} \geq \ldots \geq f_{jk_j}$. Finally, the valuations are additive across different goods. Formally, we let $\M$ be the set of all goods (that includes all copies of each type of good). For all $j \in [t]$, we denote by $\M_j \subseteq \M$ the subset of goods that are copies of good $j$. The value for any set $S \subseteq \M$ is given by $
    f(S) = \sum_{j \in [t]} \sum_{k \leq |S \cap \M_j|}f_{jk}$.

Throughout this paper, when we refer to fair division problems in the specific context of $\SPLC$ valuations, we denote the instance by $\SPLCins$.
\paragraph{Submodular Valuations.}
Submodular valuations are a popular class of valuations in the complement-free hierarchy. The valuations are characterized by the property of decreasing marginal utility. In particular, a valuation function $f(\cdot)$ is said to be submodular if and only if, for all goods $g \in \M$ and any subsets $S \subset Q \subseteq \M$, $
    f(g \mid S) \geq f(g \mid Q)$, where $f(g \mid S)$ denotes the marginal utility of good $g$ on set $S$, i.e. $f(g \mid S) \coloneqq f(g \cup S) - f(S)$.
\subsection{Fairness Notions}
\paragraph{Maximin Share ($\MMS$)}\label{sec:mms-prelims}
Maximin Share ($\MMS$) is defined for symmetric agents, that is for all agents $i \in [n]$, $b_i = \sfrac{1}{n}$. Consider a symmetric fair division instance $\MMSinssym$. The $\MMS$ value of an agent $i \in \N$ is defined as,
\begin{equation}\label{def:mms}
    \MMS_i^n([m]) \coloneqq \max_{(A_1, \ldots, A_n) \in \Pi_{[n]}([m])} \, \min_{k \in [n]} v_i(A_k)
\end{equation}
We refer to $\MMS_{i}^n([m])$ by $\MMS_i$ when the qualifiers $n$ and $m$ are clear from the context. The following claim about $\MMS$, called single good reduction is well known and we will use it in our analysis.

\begin{claim}\label{clm:single-good-redn}
    Given a fair division instance $\MMSinssym$, the $\MMS$ value of an agent is retained if we remove any single agent and any single good. That is,
        $\MMS_i^n([m]) \geq \MMS_i^{n-1}([m \setminus \{g\}])$ for all  $g \in [m]$
\end{claim}
\paragraph{Any Price Share ($\APS$)}\label{sec:aps-prelims}
Let $\Price$ denote the simplex of price vectors over the set of goods $\M = [m]$, formally, $\Price=\{(p_1,\dots,p_m)\ge 0\ |\ \sum_i p_i=1\}$. For an instance $\MMSins,$ the $\APS$ value of agent $i$ is defined as,
\begin{equation}\label{def:aps-price}
    \APS_i^n([m]) \coloneqq \min_{p \in \Price} \max_{S \subseteq [m], p(S) \leq b_i} v_i(S)
\end{equation}
where $p(S)$ is the sum of prices of goods in $S$. We will refer to $\APS_i^{[n]}([m])$ by $\APS_i$ when the qualifiers $n$ and $m$ are clear.

An alternate definition without prices is as follows. 

\begin{definition}[Any Price Share]\label{def:aps-sets}
 The $\APS$ value of an agent $i$ for an instance $\MMSins$ is the solution of the following program.
\begin{align*}
& \APS_i = \max z \\ 
    \sum_{T \subseteq [m]} &\lambda_{T} = 1 \\
    & \lambda_T = 0 &&\forall T \text{ such that } v_i(T) < z \\
     \sum_{T \subseteq [m]: j \in T} &\lambda_T \leq b_i &&\forall j \in [m] \\
    & \lambda_T \geq 0 &&\forall {T \subseteq [m]}
\end{align*}
\end{definition}
These definitions and their equivalence is stated in \cite{BabaioffEF21}. We also note the following claim that was proved in \cite{BabaioffEF21}
\begin{claim}\label{clm:mms-aps-rel}
    For any symmetric fair division instance, $\MMSinssym$, $\APS_i^n([m]) \geq \MMS_i^n([m])$ for all agents $i \in [n]$.
\end{claim}
Similar to Claim \ref{clm:single-good-redn}, single good reduction also holds for $\APS$ when agents are symmetric.
\begin{restatable}{claim}{clmsinglegoodrednaps}
\label{clm:single-good-redn-aps}
    Given a symmetric fair division instance $\MMSinssym$, the $\APS$ value of an agent is retained if we remove any single agent and any single good. That is,
        $\APS_i^n([m]) \geq \APS_i^{n-1}([m \setminus \{g\}])$ for all $g \in [m]$
\end{restatable}
\begin{proof}
Consider the dual definition of $\APS$ based on sets, definition \ref{def:aps-sets}. Consider a fair division instance, $\MMSinssym$ and fix an agent $i \in \N$. Let a solution that achieves the maximum in \ref{def:aps-sets} for $i$ be $\mathcal{S} = \{S^*_1, \ldots, S^*_r\}$ and corresponding weights on the sets be $\lambda^*_{S_j}$ for $j \in [r]$. Suppose we remove a good $g$ from and some agent $i' \neq i$ to get a new instance. The weight of $i$ in this new instance is $\frac{1}{n-1}$. Suppose $g$ belongs to some sets in $\mathcal{S}$. For all $\lambda_{S}$ such that $S \cap \{g\} \ \emptyset$, we modify $\lambda'_{S} = \lambda_{S} \cdot \frac{n}{n-1}$. We set all other $\lambda'_S$ to be zero. Clearly, we are only putting non-zero weight on sets that have value at least $\APS_i$. For any good $g' \in \M \setminus \{g\}$, 
\begin{align*}
\sum_{S : g' \in S} \lambda'_{S} &= \sum_{S : g' \in S} \lambda_{S} \cdot \frac{n}{n-1} \\
&\leq \frac{n}{n-1} \cdot \frac{1}{n} &&(\because\text{ original $\lambda_S$ are feasible)}\\
&\leq \frac{1}{n-1}.
\end{align*}
Finally, we prove that we put a total weight of $1$ on the sets. Since in original instance, the total weight on sets that included $g$ cannot exceed $\frac{1}{n}$, the total weight on other sets, aka the sets we put non-zero weight on is at least $1 - \frac{1}{n} = \frac{n-1}{n}$. As we have scaled all weights by $\frac{n}{n-1}$, the new total weight is at least $1$.
\end{proof}

\paragraph{$\alpha$-Approximate Allocations.} Given a symmetric fair division instance, $\MMSinssym$, we say that an allocation $A = (A_1, \ldots, A_n)$ is $\alpha$-$\MMS$ if for all $i \in [n]$, $v_i(A_i) \geq \alpha \MMS_i$. Similarly, for an instance $\MMSins$, we say an allocation is $\alpha$-$\APS$ if for all agents $i \in \N$, $v_i(A_i) \geq \alpha \APS_i$. The notion of $\alpha$-$\MMS$ or $\alpha$-$\APS$ allocations also applies to fractional allocations.
\subsection{Concave Extension}
The valuation functions of agents are discrete set functions. In Section \ref{sec:upper-bounds} we discuss new upper bounds on $\MMS$ and $\APS$ values of agents. These upper bounds require us to consider fractional allocations. Therefore, we must extend the input valuation function to include values for fractional sets. We assume $f:2^V\rightarrow \mathbb{R}_+$ is a non-negative monotone real-valued set function over a finite ground set $V$. A natural way to extend $f$ to a continuous concave function over the hypercube $[0,1]^V$ is called the concave closure or concave extension, and is denoted by $f^+$. It is defined as follows.
\begin{definition}[Concave extension of $f$]\label{def:concave-ext}
\begin{multline}
f^+(x) = \max_{(\alpha_S)_{S \subseteq 2^{\M}}} \{ \sum_{S \subseteq V} f(S) \alpha_S\ | \ \sum_{S} \alpha_S = 1 \text{~and~} \sum_{S \ni i} \alpha_S = x_i \quad \forall i \in V \text{~and~} \alpha_S \ge 0 \quad \forall S\subseteq V\}.
\end{multline}
\end{definition}
\begin{remark}
For submodular valuations, concave extensions $\classNP$-hard to evaluate. Therefore, another extension called the Multilinear Extension is more widely used in Fair Division literature with these functions. For those familiar with this theory, we show in Appendix \ref{app:prelims} that under Multilinear Extension with submodular valuations, there are instances where no fractional $1$-$\MMS$ allocation exists. This is the main reason we work with concave extension here.
\end{remark}
\subsection{Market (Competitive) Equilibrium}
The theory of market equilibrium typically considers fractional allocations of goods. Therefore, in this part we will assume that agents have continuous valuation function. Consider an instance $(\N, \M, \widehat{v}_i(\cdot), (b_i)_{i \in \N})$ where, for $m=|\M|$,  $\widehat{v}_i: \mathbb{R}^{m} \rightarrow \mathbb{R}_+$ is a non-negative non-decreasing continuous valuation function of agent $i$. A market equilibrium constitutes prices $\vecp = (p_j)_{j \in \M}$ of goods and an allocation $\vecx = (x_{ij})_{i \in \N, j \in \M}$, where $x_{ij}$ is the amount of good $j$ allocated to agent $i$, such that the following conditions are satisfied \cite{AD}
\begin{enumerate}
    \item   Each agent $i \in \N$ is allocated her {\em optimal  bundle}, i.e., $(x_{i1},\dots,x_{im}) \in \argmax_{{y\in \mathbb R_+^{m}: \Sigma_{j\in \M}: y_j p_j \le b_i}} \widehat{v}_i(y)$.  
    \item Market clears: all goods $j \in \M$ are completely sold, i.e. $\sum_{i \in \N} x_{ij} = 1$, or if $\sum_{i \in N} x_{ij} < 1$ then $p_j=0$.
\end{enumerate}

%% file: 3.upperbounds.tex
\section{Upper Bounds via Market Equilibrium and Concave Extension} 
\label{sec:upper-bounds}
In this section we give new upper bounds that we can use to approximate $\APS$ and $\MMS$. These upper bounds are on $\APS$ value of agents and by Claim \ref{clm:mms-aps-rel} they also upper bound the $\MMS$ value. Section \ref{sec:ub-mkts} finds an upper bound using Market Equilibrium and section \ref{sec:ub-concave-ext} finds an upper bound using concave extensions of functions. They are based on primal and dual definitions of $\APS$ respectively.
\subsection{Market Equilibrium Based 
Bound}\label{sec:ub-mkts}
Consider a fair division instance, $\MMSins$. Let $\widehat{v}_i(\cdot)$ be any continuous extension of $v_i(\cdot)$ \emph{under which market equilibrium exists}. Then we can prove the following lemma.
\begin{restatable}{lemma}{lemubmkts}
\label{lem:ub-mkts}
    For any fair division instance, $\MMSins$, let $(\vecx^*, \vecp^*)$ denote a market equilibrium with continuous extension $\widehat{v}_i(\cdot)$. Then, for all $i \in \N$, $\APS_i \leq \widehat{v}_i(\vecx_i^*)$.
\end{restatable}
\begin{proof}
    Consider Definition (\ref{def:aps-price}) of $\APS$ based on prices. We get,
    \begin{align*}
        \APS_i &= \min_{p \in \Price} \max_{S \subseteq [m], p(S) \leq b_i} v_i(S) \\
        &\leq \max_{S \subseteq [m], \vecp^*(S) \leq b_i} v_i(S) \\
        &\leq \widehat{v}_i(\vecx_i)
    \end{align*}
    where the second inequality follows because $\vecp^*$ is one particular set of prices, therefore the minimum over all prices can be at most the value at these prices. The third inequality follows because $\vecx_i$ is a fractional optimal set that agent can get under prices $\vecp^*$ and every optimal integral bundle is a potential feasible solution to the fractional optimal.
\end{proof}
\begin{remark}
For any monotone, non-negative set function, $v(\cdot)$, its concave extension $v^+(\cdot)$ will be non-decreasing, non-negative, continuous, and concave. As long as agents together are non-satiated for the available supply of goods, {\em i.e.,} have non-zero marginal up to consuming all the goods, a market (competitive) equilibrium is known to exist \cite{AD}. This will give us an upper bound to work with.
\end{remark}
\subsection{Concave Extension Based 
Bound}\label{sec:ub-concave-ext}
For a given set function $f:2^V \rightarrow \mathbb{R}_+$ and a real value $\gamma \ge 0$ we define the truncation of $f$ to $\gamma$, denoted by $\truncate{f}{\gamma}$, as follows: $\truncate{f}{\gamma}(A) = \min\{f(A), \gamma\}$ for each $A \subseteq V$. It is well-known and easy to verify that truncation of a monotone submodular function yields another monotone submodular function. We make a connection between APS value and the concave closure via the following lemma.
\begin{restatable}{lemma}{lemapsconcaveclosure}
\label{lem:aps-concave-closure}
Consider a fair division instance, $\MMSins$. For all agents $i \in \N$, $\APS_i = \sup \{ z: v_{i\downarrow z}^+(b_i,b_i,\ldots,b_i) = z\}$.
\end{restatable}
\begin{proof}
    We first show that $\APS_i \leq \sup \{ z: v_{i\downarrow z}^+(b_i,b_i,\ldots,b_i) = z\}$. Consider  $v_{i\downarrow \APS_i}^+(b_i, \ldots, b_i)$. By the dual Definition \ref{def:aps-sets} of $\APS_i$, there exist sets $\mathcal{S} = \{S_1, \ldots, S_r\}$ such that $v_i(S_j) \geq \APS_i$ for all $j \in [r]$. Therefore, $v_{i \downarrow \APS_i}(S_j) = \APS_i$ for all $j \in [r]$. The program of Definition $\ref{def:aps-sets}$ also associates variables $\lambda_{S_j}$ with these sets. Setting $\alpha_{S_j} = \lambda_{S_j}$ for all $S_j \in \mathcal{S}$ and $\alpha_{S_j} = 0$ for all other sets is a feasible solution to the program in Definition\ref{def:concave-ext}. Since value of each set is $\APS_i$, the objective value is $\APS_i$. Thus, this direction is proved.

    We now prove $\APS_i \geq \sup \{ z: v_{i\downarrow z}^+(b_i,b_i,\ldots,b_i) = z\}$. Let $z^* = \sup \{ z: v_{i\downarrow z}^+(b_i,b_i,\ldots,b_i) = z\}$. Then there exist sets $S_1, \ldots, S_r$ and values $\alpha_{S_j}$ for $j \in [r]$ that satisfy Definition \ref{def:concave-ext}. Since the function is truncated at $z^*$ and $\sum_{j \in [r]} \alpha_{S_j} v_i(S_j) = z^*$ with all $\alpha_{S_j} \leq 1$, we have that $v_i(S_j) = z^*$ for all $j \in [r]$. Thus, these sets $\{S_1, \ldots S_r\}$ and corresponding $\alpha_{S_j}$'s satisfy the constraints in Definition \ref{def:aps-sets} of $\APS_i$. Thus every feasible solution to the constraint set in Defintion \ref{def:concave-ext} is a feasible solution for the $\sf LP$ based definition of $\APS$, Definition \ref{def:aps-sets}. This implies that $\APS_i \geq \sup \{ z: v_{i\downarrow z}^+(b_i,b_i,\ldots,b_i) = z\}$ proving the claim.
\end{proof}

 \begin{remark}
     Note that the truncation is important in the indivisible setting, for otherwise the relaxation is too weak. We also see that $\APS_i$ can be computed if one can evaluate the concave closure of the truncation of $v_i$
 \end{remark}

%% file: 4.splc.tex
\section{$\nicefrac{1}{2}$-$\MMS$ for $\SPLC$ valuations}\label{sec:splc}

In this section, we give a polynomial time algorithm for computing an allocation that gives each agent a bundle they value at least half as much as their $\MMS$ value. As mentioned in the Introduction, all our results in this section hold for symmetric $\APS$ also. The section is organized as follows. In Section \ref{sec:frac-1-mms} we give a linear relaxation for $\SPLC$ valuations and show that under this relaxation, there exists a fractional allocation that gives each agent their $\MMS$ value. In Section \ref{sec:lp-gap-round} we give a linear program and show that if the program is feasible, we can, in polynomial time find an integral solution where each agent loses at most one good. Then in Section \ref{sec:half-mms-exist} we use the results in Sections \ref{sec:frac-1-mms} and \ref{sec:lp-gap-round} to prove existence of $\frac{1}{2}$-$\MMS$ allocation under $\SPLC$ valuations. Finally in Section \ref{sec:alg-half-mms}, we address the computational aspects and give a polynomial time algorithm to compute the $\frac{1}{2}$-$\MMS$ allocation.
\subsection{Fractional $1$-$\MMS$ allocation exists}\label{sec:frac-1-mms}
Recall that under $\SPLC$ valuations, we have $t$ goods and for each good, $j \in [t]$ there are $k_j$ copies. Every agent $i \in [n]$ has a value $v_{ijk}$ associated with $k^{th}$ copy of $j^{th}$ good. For an $\SPLC$ valuation $v_i(\cdot)$, we denote its linear extension by $v_i^L(\cdot)$ and define it as follows.
\begin{definition}[Linear Extension for $\SPLC$ function]
    Let variables $x_{ijk}$ denote the fraction of $k^{th}$ copy of good $j$ that agent $i$ receives. We use $\vecx_i = (x_{ijk})_{j \in [t], k \in [k_j]}$ to denote the vector of fractional allocation received by agent $i$. The value of the set is then defined as
    \begin{equation}
        v_i^L(\vecx_i) \coloneqq \sum_{j \in [t]} \sum_{k \in k_j} v_{ijk} x_{ijk}
    \end{equation}
\end{definition}
It is known from \cite{SPLC1} that a market equilibrium exists under this linear relaxation. Combining this with Lemma \ref{lem:ub-mkts}, we immediately get the following lemma.
\begin{lemma}\label{lem:exist-frac-mms}
    Given an $\SPLC$ fair division instance $\SPLCins$, there exists a fractional allocation $\vecx = (x_{ijk})_{i \in [n], j \in [t], k \in [k_j]}$ such that $v_i^L(\vecx_i) \geq \MMS_i$.
\end{lemma}
The above lemma proves existence of $1$-$\MMS$ fractional allocations under linear relaxation. However, market equilibrium under this relaxation is $\classPPAD$-hard to compute \cite{garg2012complementary}. In the next section we give a linear program that uses Lemma \ref{lem:exist-frac-mms} to compute a (fractional) $1$-$\MMS$ allocation in polynomial time.
\subsection{Linear Program to Compute a Fractional $1$-$\MMS$ allocation} \label{sec:lp-gap-round}
Given an $\SPLC$ fair division instance, $\SPLCins$, consider the following Linear Feasibility Program parameterized by $(\mu_i)_{i \in \N}$.

\begin{equation}\label{lp:splc}
\large \begin{array}{ccc}
    \sum_{j} \sum_{k} v_{ijk} x_{ijk}  \geq   \mu_i & \text{for all } i \in \N \\ \\
    \sum_{i} \sum_{k} x_{ijk} \leq  k_j  & \text{for all } j \in [t]\\ \\
    0  \leq x_{ijk}  \leq 1 & \text{for all } i, j, k
\end{array}
\end{equation}
We say that $(\mu_i)_{i \in \N}$ are \emph{feasible} if there exists a feasible solution to (\ref{lp:splc}) for the given $(\mu_i)_{i \in \N}$.

The rest of this section is dedicated to proving the following lemma which says that if $\sf LP$ (\ref{lp:splc}) is feasible, we can get an integral allocation where each agent gets a value of at least $\mu_i - {\sf Max}_i$ where ${\sf Max}_i$ is the maximum value that agent $i$ has for any good, i.e. ${\sf Max}_i = \max_{j,k} v_{ijk}$.
\begin{lemma}\label{lem:main-splc}
    Given $(\mu_i)_{i \in \N}$ that are feasible, in polynomial time we can find an integral allocation such that for all $i \in \N$, agent $i$ gets a set $A_i$ with $v_i(A_i) \geq \mu_i - \max_{j,k} v_{ijk}$.
\end{lemma}
\begin{remark}
    The rest of this section proves Lemma \ref{lem:main-splc}. Readers not wanting to go into technical details of this proof can move to Section \ref{sec:half-mms-exist} to see how to use it for proving the existence of $\frac{1}{2}$-$\MMS$ allocations.
\end{remark}
To prove Lemma \ref{lem:main-splc}, we need to start with a fractional solution that satisfies some properties enabling us to round it. Towards this, we consider the fractional solution that maximizes the following optimization program.
\begin{equation}\label{lp:splc-opt}
\large \begin{array}{ccc}
    \max \sum_{i, j, k} v_{ijk} x_{ijk} && \\ \\
    \sum_{j} \sum_{k} v_{ijk} x_{ijk}  \geq   \mu_i & \text{for all } i \in \N \\ \\
    \sum_{i} \sum_{k} x_{ijk} \leq  k_j  & \text{for all } j \in [t]\\ \\
    x_{ijk}  \leq 1 & \text{for all } i, j, k \\ \\
    x_{ijk} \geq 0 & \text{for all } i, j, k
\end{array}
\end{equation}
\label{sec:splc-acyclic}
Denote the optimal of this program with $\vecx^* = (x^*_{ijk})_{i \in \N, j \in [t], k \in [k_j]}$. We start with the following simple claim.
\begin{restatable}{claim}{clmlastgoodfrac}
    $\vecx^*$ can be modified so that for all $i \in \N, j \in [t]$, $x^*_{ijk}$ is fractional for exactly one $k$.
\end{restatable}
\begin{proof}
    For any agent $i \in [n]$, for $k < k'$ if we have $x^*_{ijk'} \neq 0$ and $x^*_{ijk} < 1$, we can change $x^*_{ijk} = \min\{1, x^*_{ijk} + x^*_{ijk'}\}$ and $x^*_{ijk'} = \max\{0, x^*_{ijk'} - (1 - x^*_{ijk})\}$. It is easy verify that this does not affect the feasibility of the program or the value of optimal solution.
\end{proof}
Therefore, for each agent and each good, we have only one copy fractionally allocated\footnote{Note that in the fractional allocation, there might be more than one copy of a good that is fractionally allocated among the agents, however each agent is only allocated one copy fractionally.}. This lets us define the following graph of fractional allocations.
\begin{definition}[Fractional Allocation Graph]\label{def:frac-graph}
   This is a bipartite graph with agents on one side and one single copy of each good on the other side. We draw edge between agent $i$ and good $j$ if there is a copy of good $j$ that agent $i$ gets fractionally. The weight of this edge is the fractional amount of the good assigned to the agent.
\end{definition}

The optimal solution obtained might be such that the Fractional Allocation Graph has cycles in it. We prove the following Lemma that says even if the graph has cycles, they can be eliminated without reducing the value received by any agent.
\begin{restatable}{lemma}{lemsplcacyclicityvalue}
\label{lem:splc-acyclicity-value}
    Given any fractional optimal solution to $\sf LP$ \emph{(\ref{lp:splc-opt})} $\vecx^*$, we can get a fractional solution, $\Bar{\vecx}$ such that the Fractional Allocation Graph corresponding to $\Bar{\vecx}$ has no cycles and $v^L_i(\vecx^*_i) = v^L_i(\Bar{\vecx}_i)$ for all $i \in \N$.
\end{restatable}
To prove the above lemma, we first introduce the concept of Fractional Price Allocation Graphs using the Lagrangian dual variables. For each good $j$ and agent $i$, let $r^i_j$ be the last copy of good $j$ that agent $i$ gets and gets it fractionally. We look at the Lagrangian corresponding to $\sf LP$ (\ref{lp:splc-opt}). Let $\beta_i$, $i \in \N$ be the dual variables corresponding to first set of constraints, $p_j$, $j \in [m]$ be the dual variables corresponding to the second set of constraints, $\lambda_{ijk}$ be the dual variables corresponding to third set of constraints and $\gamma_{ijk}$ be the variables corresponding to the non-negativity constraints. Therefore, the Lagrangian is.
    \begin{multline}\label{eqn:lp-lagrange}
        -\sum_{i, j, k} v_{ijk}x_{ijk} + \sum_{i} \beta_i \sum_{j,k} (-v_{ijk}x_{ijk} + \mu_i) + \sum_j p_j (\sum_{i,k} x_{ijk} - k_j) + \sum_{i,j,k}\lambda_{ijk}(x_{ijk} - 1) - \sum_{i,j,k}\gamma_{ijk}x_{ijk}
    \end{multline}
Let $(p_j^*)_{j \in [t]}$, $\beta_i^*$, $i \in \N$, $\lambda^*_{ijk}$ and $\gamma^*_{ijk}$ be the values for dual variables at optimal. With this notation, we first prove the following property.\footnote{For the readers familiar with Fisher markets,this property is reminiscent of the maximum bang per buck allocation in linear case.}
\begin{restatable}{lemma}{lemmbblastgood}
\label{lem:mbb_last_good}
    For all goods $j \in [m]$ that have any copy fractionally allocated, we have
    \begin{equation*}
        \frac{v_{ijr^i_j}}{p^*_j} = \frac{v_{ij'r^i_{j'}}}{p^*_{j'}}.
    \end{equation*}
\end{restatable}
\begin{proof}
    Differentiating the Lagrangian (\ref{eqn:lp-lagrange}) with respect to $x_{ijk}$, we get
    \begin{equation*}
        -v_{ijk} - \beta_i v_{ijk} + p_j + \lambda_{ijk} = \gamma_{ijk}.
    \end{equation*}
    Due to non-negativity constraints we have, $\gamma_{ijk} \geq 0$. Therefore we get,
    \begin{equation}\label{eqn:mbb}
        v_{ijk}(1 + \beta_i) \leq p_j + \lambda_{ijk}
    \end{equation}
    We can assume without loss of generality that for all $j \in [t]$, there exist $i, k$ such that $v_{ijk} > 0$. Therefore, the left hand side of the above equation is strictly greater than $0$ for some $i$, implying that $p_j + \lambda_{ijk} > 0$. Further,  by complementary slackness, $\gamma^*_{ijk} = \lambda^*_{ijk} = 0$ if $0 < x^*_{ijk} < 1$. Therefore, for any good that has some copy fractionally allocated, we have $p^*_j > 0$. This implies that for all $j \in [t]$ that have any copy fractionally allocated, we can rewrite equation (\ref{eqn:mbb}) as
    \begin{equation*}
        \frac{v_{ijk}}{p^*_j} = \frac{1}{1+\beta_i}.
    \end{equation*}
    In particular, for the last copy of each good that gets assigned fractionally, the ratio of value to price of the good is $\frac{1}{1+\beta_i}$ which is independent of the good. This proves the claim.
\end{proof}
We can now define a Fractional Price Allocation graph based on the dual variables, $(p_j)_{j \in [t]}$ that will help us to prove that the Fractional Allocation Graph can be made acyclic.
\begin{definition}[Fractional Price Allocation Graph]\label{def:frac-price-graph}
     This is a bipartite graph with agents on one side and goods on another. There is an edge between an agent, $i$ and good $j$ if the agent buys a copy of the good fractionally. The weight of the edge is $w_{ij} = x_{ij}p_j$ where $p_j$ is the dual variable defined above and $x_{ij}$ is the fractional amount of (last copy of) good $j$ that $i$ received.
\end{definition}
The following observation about the two graphs is easy to verify.
\begin{claim}
    For a given fractional solution $\vecx$, there are cycles in the Fractional Allocation Graph corresponding to it if and only if there are cycles in the Fractional Price Allocation Graph corresponding to $\vecx$.
\end{claim}
\begin{restatable}{lemma}{lemsplcacyclicity}
\label{lem:splc-acyclicity}
    Any cycles in the Fractional Price Allocation Graph can be eliminated without changing the total weight on edges incident any of the agents or any of the goods.
\end{restatable}
\begin{proof}
Consider any cycle in the Fractional Price Allocation graph. This cycle alternates between agents and goods. Let the cycle be $i_1 - j_1 - i_2 - j_2 - \cdots - j_{\ell-1} - i_{\ell} - j_{\ell} - i_1$. To eliminate this cycle, we find the edge with minimum weight. Without loss of generality, assume that this edge $\{i_1, j_1\}$.\footnote{We can always relabel the cycle so that this is true.} Denote the weight of this edge, $p_{i_1j_1}$ by $\delta$. We now change the weights on the edges to be:
\begin{align*}
    w_{i_{k}j_{k}} &= w_{i_{k}j_{k}} - \delta  &&\text{ for all } k \in [\ell] \\
     w_{i_{k}j_{k-1}} &= w_{i_{k}j_{k}} + \delta  &&\text{ for all } k \in [\ell]
\end{align*}
where we define $j_0$ to be $j_{\ell}$.
Thus, in the price allocation graph, we eliminate the edge $i_1j_1$. Further note that for each agent and each good, we add $\delta$ on one edge and remove $\delta$ on the other edge. Therefore, we don't change the total weight incident on any agent or any good. This completes the proof.
\end{proof}
We can now prove Lemma \ref{lem:splc-acyclicity-value}.
\begin{proof}[Proof of Lemma \ref{lem:splc-acyclicity-value}]
    Consider the Fractional Price Allocation Graph corresponding to $\vecx^*$. Using Lemma \ref{lem:splc-acyclicity}, get a new fractional solution $\Bar{\vecx}$ such that the Fractional Price Allocation Graph corresponding to $\Bar{\vecx}$ has no cycles and the total weight incident on any agent is unchanged. Let $F_i$ denote the set of all goods that agent $i$ has received fractionally. We show that the total value received by the agent remains unchanged in $\Bar{\vecx}$.
    \begin{align*}
        v_i^L(\vecx^*) &= \sum_{j \in F_i} v_{ij} x^*_{ij} \\ 
        &= \sum_{j \in F_i} \frac{w_{ij}}{p^*_j} v_{ij} && \text{(By definition of }w_{ij}) \\
        &= \frac{1}{1+\beta_i} \sum_{j \in F_i} w_{ij}
    \end{align*}
    Since the sum of weights incident is same in $\vecx^*$ and $\Bar{\vecx}$, the final value is same for $\vecx^*$ and $\Bar{\vecx}$. This proves the lemma.
\end{proof}

Algorithm \ref{alg:round} describes the rounding algorithm that lets us prove Lemma \ref{lem:main-splc}.
\begin{algorithm}[tbh!]
\caption{Rounding the Fractional Optimal of $\sf LP$ (\ref{lp:splc-opt})}
\label{alg:round}
\DontPrintSemicolon
  \SetKwFunction{Define}{Define}
  \SetKwInOut{Input}{Input}\SetKwInOut{Output}{Output}
  \Input{$\vecx^*$ that is an optimal solution to $\sf LP$ (\ref{lp:splc-opt}).}
  \Output{Integral Allocation $A = (A_1, \ldots, A_n)$ where each agent $i \in \N$ receives a value of at least $\mu_i - \max_{j,k}v_{ijk}$.}
  \BlankLine
  Use Lemma \ref{lem:splc-acyclicity-value} to convert $\vecx^*$ to $\Bar{\vecx}$ such that the Fractional Allocation Graph corresponding to $\Bar{\vecx}$ has no cycles and value of each agent is same.\;
  Create Fractional Allocation Graph corresponding to $\Bar{\vecx}$.\label{step:cycle-elim}\;
  In the forest obtained, root every tree by an arbitrary agent.\;
  For each tree, assign all the children goods to the parent agent.\label{step:bottom-up-round}\;
\end{algorithm}
\begin{proof}[Proof of Lemma \ref{lem:main-splc}]
    Consider the rounding procedure described in Algorithm \ref{alg:round}. From Lemma \ref{lem:splc-acyclicity-value}, we get that the agents have lost no value up to step \ref{step:cycle-elim}. After that in Step \ref{step:bottom-up-round}, every agent loses at most one good -- the parent good in the tree. This proves the lemma's claim on value of agent. To see that the algorithm runs in polynomial time, note that solving a Linear Program can be done in polynomial time using any of the well known algorithms. Cycle cancellation removes a cycle by deleting one edge. Since there are at most $n \cdot t$ edges in the graph, this also runs in polynomial time. Therefore, overall the subroutine runs in polynomial time.
\end{proof}
\subsection{Existence of $\frac{1}{2}$-$\MMS$ Allocations}\label{sec:half-mms-exist}
Algorithm \ref{alg:splc-mms} gives an algorithm that proves the existence of $\frac{1}{2}$-$\MMS$ allocations. It does so by assuming we can compute $\MMS$ values for all agents. While these are $\classNP$-hard to compute making the algorithm non-polynomial time, it exhibits all other important ideas in getting a polynomial time $\frac{1}{2}$-$\MMS$ allocation.
\begin{algorithm}[tbh!]
\caption{Existence of $\frac{1}{2}$-$\MMS$ Algorithm for  $\SPLC$ valuations}\label{alg:splc-mms}
\DontPrintSemicolon
  \SetKwFunction{Define}{Define}
  \SetKwInOut{Input}{Input}\SetKwInOut{Output}{Output}
  \Input{$\SPLCins$}
  \Output{Allocation $A$ where for every agent $v_i(A_i)\ge \MMS_i/2$}
  \BlankLine
  Define $\mu_i \coloneqq \MMS_i$ for all $i \in \N$.\;
  Initialize set of \emph{active} agents $\A = \N$ and set of \emph{active} goods $\G = \M$.\;
  \While{there exists $i \in \N$ and $j \in [t]$ such that $v_{ij1} \geq \frac{\mu_i}{2}$} {\label{step:single-good-start}
    $A_i \leftarrow \{j\}$\;
    $\A \leftarrow \A \setminus \{i\}$\;
    Remove one copy of $j$ from $\G$\;
  }\label{step:single-good-end}
  Solve the Linear Program (\ref{lp:splc-opt}) for the reduced instance $(\A, \G, (v_i)_{i \in \N})$ and obtain fractional allocation $\vecx$.\;
  Use Lemma \ref{lem:main-splc} to get an integral allocation and assign appropriate bundles to agents in $\A$.\;
  Return $A = (A_i)_{i \in \N}$.
\end{algorithm}
\begin{theorem}\label{thm:half-mms-exist}
    Given an $\SPLC$ fair division instance, $\SPLCins$, Algorithm $\ref{alg:splc-mms-complete}$ gives each agent a bundle $A_i$ such that $v_i(A_i) \geq \frac{1}{2} \MMS_i$.
\end{theorem}
\begin{proof}
    Note that for any agent who gets allocated an item in Steps $\ref{step:single-good-start}$ to $\ref{step:single-good-end}$, they get a value of $\frac{\MMS_i}{2}$ by definition of $\mu_i = \MMS_i$. Recall from Claim \ref{clm:single-good-redn} that removal of one good and one agent retains the $\MMS$ value of the agent. Therefore, for all remaining agents, at the end of while loop from Step \ref{step:single-good-start} to \ref{step:single-good-end}, the $\MMS$ value of active agents is retained in the reduced instance. Further the reduced instance satisfies the property that for each agent $i \in \N$, $\max_{j,k} v_{ijk} \leq \frac{\MMS_i}{2}$. Finally, using Lemma \ref{lem:exist-frac-mms}, note that $\mu_i \coloneqq \MMS_i$ for all $i \in \A$ (the set of Active agents left after single good reduction), form \emph{feasible} $(\mu_i)_{i \in \A}$. Therefore, we can use Lemma \ref{lem:main-splc} to get an integral allocation where each agent, $i \in \A$ gets a value of $v_i(A_i) \geq \MMS_i - \max_{j,k} v_{ijk} \geq \frac{\MMS_i}{2}$. This completes the proof.
\end{proof}

\subsection{Algorithm : Computational Aspects}\label{sec:alg-half-mms}
We now address the final computational aspect of finding $\MMS_i$ values here. Since they are $\classNP$-hard to compute, we use $\mu_i$ values as follows:
\begin{equation}\label{eqn:mu_i}
    \mu_i = \sum_{j \in [t]} \left(\sum_{k \leq \frac{k_j}{n}}v_{ijk} + v_{ij{\lceil\sfrac{k_j}{n} \rceil}}\left(\frac{k_j}{n} - \left\lfloor \frac{k_j}{n} \right\rfloor \right)\right)
\end{equation}
In Appendix \ref{app:splc}, we give the complete algorithm that shows how to use these $\mu_i$ to get a $\frac{1}{2}$-$\MMS$ allocation. We only state the main theorem here.

\begin{restatable}{theorem}{algsplcmmscomplete}
    There exists an algorithm which in polynomial time outputs an allocation that gives $\frac{1}{2}$-$\MMS$ to each agent.
\end{restatable}
\begin{remark}
    The values $\mu_i$ defined in Equation \ref{eqn:mu_i} are upper bounds on $\APS$ also. From Claim \ref{clm:single-good-redn-aps}, single good reduction also holds for $\APS$ with symmetric agents. These two things together give us that all our results of this section also go through for $\APS$ approximation for symmetric agents.
\end{remark}

%% file: 5.submodular.tex
\section{Approximate APS allocations with Submodular valuations}\label{sec:sub-aps}
We now move to the more general case of submodular valuations and agents with non-symmetric entitlements. In this section, given a fair division instance, $\MMSins$, we show that a greedy algorithm gives us an allocation $A = (A_1, \ldots, A_n)$ such that $v_i(A_i) \geq \frac{1}{3} \APS_i$ for all $i \in \N$. While this algorithm is simpler than the Algorithm for $\SPLC$ valuations seen in previous section, we give an example in Appendix \ref{app:intro} (example \ref{eg:greedy-half-splc-counter}) to show that a natural modification fails to give a $\frac{1}{2}$ approximation for even $\SPLC$ valuations. The section is organized as follows. In Section \ref{sec:aps-props}, we outline some basic properties for $\APS$. Section \ref{sec:alg} describes the details of the greedy algorithm assuming we can compute the $\APS$ values for agents. Section \ref{sec:analysis} analyses the algorithm and shows how to bypass the computational aspect of computing $\APS$ values.
\subsection{Properties of $\APS$}\label{sec:aps-props}

We will use the following properties of $\APS$ in our analysis. The proofs are simple and we omit them here.
\begin{claim}
\label{claim:aps-concave-closure}
For any $z \le APS_i$, $v_{i\downarrow z}^+(b_i,b_i,\ldots,b_i) = z$. Moreover, if $\alpha_S, S \subseteq [m]$ is a feasible solution to the LP defining the value $v_{i\downarrow z}^+(b_i,b_i,\ldots,b_i)$ then $\alpha_S > 0$ implies that $v_{i \downarrow z}(S) = z$.
\end{claim}

\begin{claim}\label{clm:aps-scale-free}[Scale-freeness]
    If a valuation function $v:2^{\M}\rightarrow \mathbb{R}_{\ge 0}$ is multiplied by any scalar value $\alpha,$ the $\APS$ value of an agent with the scaled function is $\alpha$ times the $\APS$ value of an agent with the original function and same entitlement.
\end{claim}

\begin{claim}\label{clm:aps-cap-free}[Capping retains $\APS$]
If an agent with some entitlement $b$ and valuation function $v:2^{\M}\rightarrow \mathbb{R}_{\ge 0}$ has $\APS$ value $\gamma,$ then the $\APS$ value of an agent with the same entitlement and the capped function $v'(\Set):=\min\{\APS, v(\Set)\}$ $\forall Set\subseteq \M$ is also $\gamma.$ 
\end{claim}

\subsection{Algorithm.}\label{sec:alg}
The main subroutine of our algorithm is Algorithm \ref{alg:submodular-aps}, which takes as input a fair division instance $(\N, \M, (v_i)_{i \in \N}, (b_i)_{i \in \N})$ along with a vector of numbers $(\beta_i)_{i \in \N}$. It outputs an allocation with the following property:  each agent $i$ receives a bundle of value at least $\frac13 \beta_i$ if $\beta_i \leq \APS_i$; note that the guarantee holds for agent $i$ even if the $\beta$ values for other agents are larger than their $\APS$ values. In particular, if we have the $\APS$ values of all agents, the algorithm is efficient and gives each agent $\frac{1}{3}$ fraction of her $\APS$. Note that computing $\APS$ values is known to be $\classNP$-hard. In the next section, we give Algorithm \ref{alg:submodular-aps-search} which is a wrapper around Algorithm \ref{alg:submodular-aps} and bypasses this computational issue.

First, we explain Algorithm \ref{alg:submodular-aps}. This algorithm works as follows. Given an instance of the fair division problem, and guesses for the $\APS$ values of all agents ($\beta_i$ for agent $i$), it performs two pre-processing steps. 
The intuition is the following. Suppose we knew the exact $\APS_i$ values for each $i$. Since our algorithm is greedy, we must normalise valuations. $\APS$ values are scale invariant (Claim \ref{clm:aps-scale-free}) therefore, we scale $v_i$ such that $\APS_i = nb_i$ implying $\sum_i \APS_i = n$. More over, if we knew $\APS_i$ then truncating $v_i$ to $\APS_i$ is convenient and does not affect the value (Claim \ref{clm:aps-cap-free}). In our algorithm we do have $\beta_i$ values rather than $\APS_i$ values. Nevertheless we will proceed as if these values are correct estimates. After scaling and truncating we have for each agent $i$ a valuation function $\hat{v}_i$. We have the property that $\hat{v}_i$ is truncated at $nb_i$.

The key part of the algorithm is the following greedy strategy. The algorithm allocates goods in multiple rounds.
Each round greedily chooses an agent and a good that maximizes the objective $\min\{2nb_i/3, v_i(j\mid A_i)\}$ over all agents and goods. Recall that $v_i(j\mid A_i)$ is the marginal value of the (unallocated) good $j$ to agent $i$'s current bundle $A_i$. The selected good is allocated to the selected agent. As soon as an agent receives a bundle of value at least $nb_i/3$, it is satisfied and removed from consideration in future rounds.

It is easy to see that the algorithm will terminate in at most $m$ rounds. We will show that at termination the following property is true: for each agent $i$ such that $\beta_i \le \APS_i$, $v_i(A_i) \ge \beta_i/3$ where $A_i$ is the allocation to $i$. This will be used to obtain a polynomial time algorithm for a $\frac13$-$\APS$ allocation.

\begin{algorithm}[tbh!]
\caption{Greedy Procedure for $\APS$ with Submodular Valuations}\label{alg:submodular-aps}
\DontPrintSemicolon
  \SetKwFunction{Define}{Define}
  \SetKwInOut{Input}{Input}\SetKwInOut{Output}{Output}
  \Input{$\MMSins,$ vector $(\beta_i)_{i\in\N}$}
  \Output{Either (i) allocation $\A$ where for every agent $v_i(A_i)\ge \beta_i/3,$\\ or (ii) some agent $i:$ $\beta_i>\APS_i$}
  \BlankLine
  Normalization: $v_i'\coloneqq$ normalized $v_i$ so that $\beta_i=nb_i$ that is, for all sets $S \subseteq \M$, $v_i'(S) = v_i(S) \cdot \frac{nb_i}{\beta_i}$\;
  Truncation: $\hat{v_i}\coloneqq \truncate{{v_i}}{nb_i}$ \;
  Initialize $\PA\leftarrow (\emptyset)_{i\in\N}, \M^r\leftarrow \M, \N^r\leftarrow \N$\; \tcp*{$\N^r$ is list of active agents and $\M^r$ is unallocated items}
\While{$\M^r\ne \emptyset$ and $\N^r \ne \emptyset$}{
    Let $S = \{(i, j)\mid i\in \N^r, j\in \M^r\}$\tcp*{all remaining agent-good pairs}
    $(i^*, j^*)\in \argmax\limits_{(i,j)\in S}\min\{\frac{2}{3}nb_i, \hat{v}_i(j\mid A_i)\}$ \label{greedy-step}\tcp*{greedily choose the agent with the highest marginal capped at $2nb_i/3$}
    $A_{i^*}\leftarrow A_{i^*}\cup \{j^*\}, \M^r\leftarrow \M^r\backslash \{j^*\}$\tcp*{allocate the chosen good to the agent}
    \If{$\hat{v}_{i*}(A_{i^*})\ge nb_{i^*}/3$}{
    $\N^r\leftarrow \N^r\backslash \{i^*\}$\tcp*{remove agent if they received at least $nb_i/3$}
    }
}
\If{$\exists i\in [n], v_i(A_i) < \beta_i/3$}
{\Return one such $i$}
\Return $\PA$
\end{algorithm}

\subsection{Analysis}\label{sec:analysis}
\subsubsection{Approximation Guarantee}
The following is the main theorem of our section.
\begin{restatable}{theorem}{thmmainsubaps}
\label{thm:main-sub-aps}
If $\beta_i \leq \APS_i$ for an agent $i \in \A$, then Algorithm \ref{alg:submodular-aps} terminates with an allocation such that $v_i(A_i) \geq \frac{1}{3} \beta_i$.
\end{restatable}

We prove Theorem \ref{thm:main-sub-aps} by contradiction.
Fix an agent $i$ and suppose $\beta_i \le \APS_i$ and $\hat{v}_i(A_i) < nb_i/3$. The algorithm removes $i$ from consideration during the algorithm if at time $t$, $\hat{v}_i(A_i^t) \ge nb_i/3$, thus $i$ must have stayed active until termination and the algorithm allocated all goods to agents. We compute the total sum of marginals that agent $i$ sees for all the goods (allocated to her and other agents). We compute this value in two different ways and obtain a contradiction. We now proceed with analysis of both algorithms here.

Let $A_1,A_2,\ldots,A_n$ be the allocation produced by Algorithm \ref{alg:submodular-aps}. Our goal is to prove that for each agent $i$ with $\beta_i \le \APS_i$, $v_i(A_i) \ge \beta_i/3$.
We will use the following notation. Algorithm \ref{alg:submodular-aps} irrevocably assigns one good in each iteration of the while loop. We denote the good assigned in any iteration $t$ by $j_t$, and the round in which any good $j$ was assigned by $t_j$. Let $\hat{v}_i$ denote the valuation function of agent $i$ after the pre-processing steps. Also let $A_i^t$ denote the bundle allocated to agent $i$ at the beginning of round $t$ (that is before $j_t$ is allocated). 

Now Step \ref{greedy-step} iteratively chooses an (agent, good) pair using a greedy strategy - the appropriately capped maximum marginal value of a good. We define three terms that arise from the greedy strategy definition. First, suppose the pair $(i,j)$ was chosen in some round $t.$ Let $\allocvalue(j)$ denote the unique capped marginal value that it contributed to $i$ when it was assigned, that is, $\allocvalue(j) \coloneqq \min \left\{ \frac{2}{3}nb_i, \hat{v}_i(j \mid A_i^t) \right\}.$ Second, for any round $t$, and an agent $i'$ that is active in round $t$, denote by $\marginalbid(i',t)$ the maximum (capped) marginal value that $i'$ \textit{sees}, meaning the greedy strategy value computed by fixing agent $i'$ and maximizing only over all the remaining goods; formally, $\marginalbid(i',t) =\max_{j \in \M^r} \min \left\{ \frac{2}{3}nb_{i'}, \hat{v}_{i'}(j \mid A_{i'}^t) \right\}$. Finally, for each good $j$, we denote by $\marginalbid_i(j)$ the capped marginal value that agent $i$ sees for good $j$ in the round that $j$ was allocated, that is, $\marginalbid_{i}(j) =\min \left\{ \frac{2}{3}nb_{i}, \hat{v}_{i}(j \mid A_{i}^{t_j}) \right\}$. To summarize, $\allocvalue(j)$ is the value at which good $j$ was assigned and $\marginalbid_{i}(j)$ is the value that $i$ would have received for good $j$. Because of the greedy nature of our algorithm, $\allocvalue(j) \geq \marginalbid_{i}(j)$. 

We first make a basic observation. $\hat{v}_i(S) = \min\{nb_i, \frac{nb_i}{\beta_i} v_i(S)\}$ for all $S \subseteq [m]$. Thus,
$v_i(A_i) \ge \frac{1}{3} \beta_i$ iff $\hat{v}_i(A_i) \ge nb_i/3$. Therefore, now on we will focus on proving $\hat{v}_i(A_i) \ge nb_i/3$ whenever $\beta_i \le \APS_i$, which will imply Theorem \ref{thm:main-sub-aps}. As mentioned in Section \ref{sec:sub-aps}, we prove Theorem \ref{thm:main-sub-aps} by contradiction.
Fix an agent $i$ and suppose $\beta_i \le \APS_i$ and $\hat{v}_i(A_i) < nb_i/3$. The algorithm removes $i$ from consideration during the algorithm if at time $t$, $\hat{v}_i(A_i^t) \ge nb_i/3$, thus $i$ must have stayed active until termination and the algorithm allocated all goods to agents. We compute the total sum of marginals that agent $i$ sees for all the goods (allocated to her and other agents). We compute this value in two different ways and obtain a contradiction.

\begin{lemma}\label{lem:agent-won-bids-sum}
For any agent $i \in \N$, 
$$\sum_{j \in A_i}\allocvalue(j) = \sum_{j \in A_i} \marginalbid_{i}(j) \le \frac{2}{3}nb_i.$$  
Moreover, if $|A_i| > 1$ then $\hat{v}_i(A_i) = \sum_{j \in A_i}\allocvalue(j)$.
\end{lemma}
\begin{proof}
    Consider two cases, where $|A_i|=1$ and $|A_i|>1$ In the first case, agent $i$ wins exactly one good. If $r$ is the round when $i$ wins, since the marginal of the good she gets is capped at $\frac{2}{3}nb_i$, hence by definition, $\sum_{j \in A_i} \allocvalue(j) = \allocvalue(j_r) \leq \frac{2}{3}nb_i$. 
    
    Now consider the case when $|A_i|>1$. Let $j_1,j_2,\ldots,j_h$ be the items allocated to $i$ in increasing order of time, that is $t_{j_1} < t_{j_2} \ldots < t_{j_h}$. For ease of notation we use $t_\ell$ in place of $t_{j_\ell}$.
    Note that $i$ is not allotted further goods by the algorithm once the value for their bundle is at least $nb_i/3$. This gives us two things. First, $\hat{v}_i(j_1) < nb_i/3$. Hence, by greedy choice, for any $\ell > 1$, $\hat{v}_i(j_\ell) < nb_i/3$, and by submodularity,
    $\hat{v}_i(j_\ell \mid A_{i}^{t_\ell}) < nb_i/3$.
    Therefore, the marginal value of $i$ for all the goods they win, over the empty set, thus by submodularity over any set, is smaller than $nb_i/3.$ Second, $i$'s value before they received the last good is smaller than $nb_i/3$.
    Combining these inferences we get,
    \begin{align*}
        \sum_{j \in A_i} \allocvalue(j) &= \sum_{\ell = 1}^h \min \left\{ \frac{2}{3}nb_i, \hat{v}_i(j_\ell \mid A_i^{t_\ell}) \right\} \\
        & = \sum_{\ell = 1}^h \hat{v}_i(j_\ell \mid A_i^{t_\ell}) \\
        & = \sum_{\ell = 1}^{h-1} \hat{v}_i(j_\ell \mid A_i^{t_\ell}) + \hat{v}_i(j_h \mid A_i^{t_h})\\
        &= \hat{v}_i(A_i^{t_h}) + \hat{v}_i(j_h \mid A_i^{t_h}) \\
        & \leq \frac{1}{3}nb_i + \frac{1}{3}nb_i \\
        & \leq \frac{2}{3}nb_i. \qedhere
    \end{align*}
\end{proof}

\begin{lemma}\label{lem:upper-bound}
Consider agent $i$ that is active through out the algorithm. The total sum of capped marginals over all goods in $[m] \setminus A_i$, according to agent $i$, is at most $\frac{2}{3}n(1-b_i)$, that is,
\begin{equation*}
    \sum_{j \in [m] \setminus A_k} \marginalbid_i(j) \leq \frac{2}{3}n (1-b_i).
\end{equation*}
\end{lemma}
\begin{proof}
If $i$ is active through out the algorithm then all items are allocated to some agent. Consider any item $j \in [m]\setminus A_i$. By the greedy rule $\allocvalue(j) \ge \marginalbid_i(j)$, otherwise $j$ would be allocated to $i$. Thus 
 $\sum_{j \in [m] \setminus A_i} \marginalbid_i(j) \le  \sum_{j \in [m] \setminus A_i} \allocvalue(j)$. Since each item in $j \in [m] \setminus A_i$ is allocated to some other agent, we can use Lemma~\ref{lem:agent-won-bids-sum} and see that
 $$\sum_{j \in [m] \setminus A_i} \allocvalue(j) = \sum_{k \neq i} \sum_{j \in A_{k}} \allocvalue(j) \le \sum_{k\neq i} 2nb_{k}/3 = 2n(1-b_i)/3.$$
\end{proof}

\begin{lemma}\label{lem:lower-bound}
    Suppose there is an agent $i$ with $\beta_i\le \APS_i$ and receives a bundle $A_i$ such that $\hat{v}_i(A_i) < \frac{1}{3}nb_i$ then, \begin{equation*}
        \sum_{j \in [m]\setminus A_i} \marginalbid_i(j) \ge \frac{2}{3}n.
    \end{equation*}
\end{lemma}
\begin{proof}
$\APS_i$ is the $\APS$ value for agent $i$ with valuation function $v_i$. Recall that $v'_i = \frac{nb_i}{\beta_i} v_i$ and thus the $\APS$ value with respect to $v'_i$ is
$\frac{nb_i}{\beta_i} \APS_i \ge nb_i$ ($\because \ \beta_i \le \APS_i$). The function $\hat{v}_i$ is obtained by truncating $v'_i$ at $nb_i$. Thus, from Claim~\ref{claim:aps-concave-closure}, we have $\hat{v}_i^+(b_i,b_i,\ldots,b_i) = nb_i$. From the definition of concave extension and Claim~\ref{claim:aps-concave-closure}, there are sets of items
$S_1,S_2,\ldots,S_h$ and reals $\alpha_1,\alpha_2,\ldots,\alpha_h \in [0,1]$ such that (i) $\sum_{r=1}^j \alpha_r = 1$ and (ii) $\hat{v}_i(S_r) = nb_i$ for each $r \in [h]$ and (iii) for each item $j$, $\sum_{r: j \in S_r} \alpha_r \le b_i$.
    \begin{align}
        \nonumber\sum_{j \in [m] \setminus A_i} \marginalbid_i(j) &= \sum_{j \in [m]\setminus A_i } \left( \frac{1}{b_i} \sum_{r : j \in S_r} \alpha_r \right) \marginalbid_i(j) 
        \\
        \nonumber&\qquad\qquad\text{(by property (iii) above.)} \\
        \nonumber&= \frac{1}{b_i} \sum_{j \in [m] \setminus A_i} \sum_{r: j \in S_r} \alpha_r \marginalbid_i(j) \\
        \nonumber&= \frac{1}{b_i} \sum_{r \in [h]} \sum_{j \in S_r \setminus A_i} \label{eqn:marg_sums_sets}\alpha_r \marginalbid_i(j) \\
        \nonumber&\qquad\qquad\text{(changing the order of summations.)} \\
        &= \frac{1}{b_i} \sum_{r \in [h]} \alpha_r \sum_{j \in S_r \setminus A_i}  \marginalbid_i(j)
    \end{align}
    Now by definition of Step \ref{greedy-step},
    \begin{align*}
        \sum_{j \in S_r \setminus A_i} \marginalbid_i(j) = \sum_{j \in S_r \setminus A_i} \min \left\{ \frac{2}{3}nb_i, \hat{v}_i(j \mid A_i^{t(j)})\right\}
    \end{align*}
    If any of the minimums in the above equation is $\frac{2}{3}nb_i$, then $\sum_{j \in S_r} \marginalbid_i(j) \geq \frac{2}{3}nb_i$. On the other hand, suppose all minimums evaluate to the uncapped marginal value. Then we get,
    \begin{align*}
        \sum_{j \in S_r \setminus A_i} \marginalbid_i(j) &= \sum_{j \in S_r \setminus A_i} \hat{v}_i(j \mid A_i^{t(j)})\\
    &\geq \sum_{j \in S_r \setminus A_i} \hat{v}_i(j \mid A_i) &&\text{(by submodularity.)}
        \end{align*}
    Note that $\forall j \in A_i, \hat{v}_i(j\mid A_i)=0$. Therefore,
        \begin{align*}
        \sum_{j \in S_r \setminus A_i} \marginalbid_i(j) &= \sum_{j \in S_r} \hat{v}_i(j \mid A_i)\\ 
        &\geq \hat{v}_i(S_r \mid A_i) &&\text{(by submodularity.)} \\
        &=  \hat{v}_i(S_r \cup A_i) - \hat{v}_i(A_i) &&\text{(by definition.)} \\
        &\ge  \hat{v}_i(S_r) - \hat{v}_i(A_i) &&\text{(by monotonicity.)} \\
        &> nb_i - \frac{1}{3}nb_i = \frac{2}{3}nb_i
    \end{align*}
    Substituting in equation \ref{eqn:marg_sums_sets}, we get
    \begin{equation*}
        \sum_{j \in [m] \setminus A_i} \marginalbid_i(j) \ge \frac{1}{b_i} \sum_{r \in [h]} \alpha_r \frac{2}{3}nb_i = \frac{2}{3}n. \qedhere
    \end{equation*}
\end{proof}

Now we finish the proof of Theorem \ref{thm:main-sub-aps}.
\begin{proof}[Proof of Theorem \ref{thm:main-sub-aps}]
   
Recall that we are assuming that there is an agent $i$ with $\beta_i \le \APS_i$ who received a bundle $A_i$ such that $\hat{v}_i(A_i) < nb_i/3$. Since $i$ will be present until the end of the algorithm, Lemma \ref{lem:upper-bound} implies
$\sum_{j \in [m] \setminus A_i} \marginalbid_i(j) \le \frac{2}{3}n(1-b_i)$. While Lemma \ref{lem:upper-bound} implies $\sum_{j \in [m] \setminus A_i} \marginalbid_i(j) > \frac{2}{3}n$. These two inequalities clearly contradicts each other. Therefore, it follows that every agent with $\beta_i \le \APS_i$ is guaranteed a bundle $A_i$ with $\hat{v}_i(A_i) \ge nb_i/3$ equivalently $v_i(A_i) \ge \beta_i/3$.
\end{proof}
\begin{remark}
The greedy strategy of Algorithm \ref{alg:submodular-aps} caps the marginal values to $2nb_i/3.$ The reason for adding this cap instead of considering the uncapped marginals is single good allocations. The analysis fails for the alternative simpler greedy strategy of uncapped marginals only in the first case of Lemma \ref{lem:agent-won-bids-sum}. Therefore, if the input instance satisfies $v_i(\{j\})\le 2\APS_i/3$ for every $i$ and every good $j,$ then the simpler greedy strategy will suffice. For obtaining approximate $\MMS$ allocations, we can assume this without loss of generality, by what are popularly called \textit{single good reductions}. However, such reductions do not work for the $\APS$ problem, and adding the caps cleverly bypasses this issue.
\end{remark}

\subsubsection{Computational Aspects}
Finally, we deal with the computational aspect of not having exact $\APS$ values of the agents. We give Algorithm \ref{alg:submodular-aps-search} that starts with very large guess for the $\APS$ values and iteratively decreases the guesses till we get a feasible solution.
\begin{algorithm}[tbh!]
\caption{$\frac{1}{3}$-$\APS$ Algorithm for  Submodular Valuations}\label{alg:submodular-aps-search}
\DontPrintSemicolon
  \SetKwFunction{Define}{Define}
  \SetKwInOut{Input}{Input}\SetKwInOut{Output}{Output}
  \Input{$\MMSins$}
  \Output{Allocation $\A$ where for every agent $v_i(A_i)\ge \APS_i/3$}
  \BlankLine
  Initialize $\beta_i \leftarrow v_i(\M)$ for all $i \in [n]$\;
\While{$1$}{
    \If{Algorithm \ref{alg:submodular-aps} $(\MMSins$, $(\beta_i)_{i \in [n]})$ returns an allocation $\A$}{
        return $\A$
    }
    \Else{reduce $\beta_i \leftarrow (1 - \epsilon) \beta_i$ for the $i$ returned by Algorithm \ref{alg:submodular-aps}.}
}
\end{algorithm}
\begin{lemma}\label{lem:poly-time}
For any $\epsilon>0,$ Algorithm \ref{alg:submodular-aps} computes a $\frac{1}{3(1+\epsilon)}$-$\APS$ allocation for an instance $\MMSins$ with submodular valuation functions in $O(m^2n\log_{(1+\epsilon)} v^{max}(\M))$ time.    
\end{lemma}
\begin{proof}
First, we know from Theorem \ref{thm:main-sub-aps} that if $\beta_i\le \APS_i,$ then agent $i$ will receive a bundle of value at least $\beta_i/3.$ This is independent of the $\beta$ values of the other agents and their value for their own bundles. In Algorithm \ref{alg:submodular-aps-search} we use this guarantee to guess a value of $\beta_i$ for every $i$ that is at most $1/(1+\epsilon)$ factor smaller than $\APS_i$ as follows. Note that, $\APS_i \le v_i(\M)$ for every agent $i$, hence we start at a high guess value of $\beta_i=v_i(\M).$ Now iteratively, run Algorithm \ref{alg:submodular-aps} for the current value $\beta_i$s, and if the algorithm returns agent $k$ then decrease her $\beta_k$ by a factor of $1/(1+\epsilon)$. We stop when the algorithm returns an allocation $\A$. By Theorem \ref{thm:main-sub-aps}, we know that if the algorithm returns $k$, then $\beta_k > \APS_k$. It is easy to detect agents with zero APS value, and they can be safely removed. For each remaining agent $k$, we have that $\APS_k \ge \min_{j\in \M, v_k(\{j\}) >0} v_k(\{j\})$ -- let us denote this lower bound by $v_k^{min}$. Hence, $\beta_k$ is updated in at most $\log_{(1+\epsilon)}{\frac{v_k(\M)}{v_k^{min}}}$. As there are $n$ agents, if $v^{max}(\M)$ denotes $\max_{i\in \N}\frac{v_i(\M)}{v_i^{min}},$ there are at most $n\log v^{max}(\M)$ iterations of Algorithm \ref{alg:submodular-aps-search}. Each iteration runs Algorithm \ref{alg:submodular-aps} once. This algorithm has at most $m$ iterations, as it assigns one good per iteration, and each iteration computes the maximum from $O(mn)$ values. Therefore, Algorithm \ref{alg:submodular-aps} runs in $O(m^2n)$ time. Combining together, Algorithm \ref{alg:submodular-aps-search} runs in $O(m^2n\log_{(1+\epsilon)} v^{max}(\M))$ time.
%
\end{proof}

\begin{remark}
    Theorem \ref{thm:main-sub-aps} shows existence of $\frac{1}{3}$-APS while Lemma \ref{lem:poly-time} implies an FPTAS to compute $\frac{1}{3(1+\epsilon)}$-APS allocation for any $\epsilon>0$. We can convert this FPTAS to a polynomial time algorithm to compute an exact $\frac{1}{3}$-APS by exploiting the gap of $(1-b_i)$ in the upper bound and lower bound from Lemma \ref{lem:upper-bound} and Lemma \ref{lem:lower-bound} respectively. For this, the algorithm needs to be modified a bit to ensure $v_i(A_i) \ge \frac{\beta_i}{3}(1+O(b_{min}))$ whenever $\beta_i\le \APS_i$ (similar to Theorem \ref{thm:main-sub-aps}), where $b_{min} = \min_{i\in \mathcal N} b_i$. Then, setting $\epsilon=O(b_{min})$ in Lemma \ref{lem:poly-time} will ensure $1/3$-APS allocation in $O(\frac{1}{b_{min}}m^2 n \log v^{max}(\M))$ time. Now note that, it is without loss of generality to assume that $b_i \ge \frac{1}{m},\ \forall i\in \mathcal N$, since any agent with $b_i < \frac{1}{m}$ will have $\APS_i=0$ (by Definition \ref{def:aps-price}) and hence can be discarded up front. Thus, we have $b_{min} \ge \frac{1}{m}$, and thereby we get running time of $O( m^3 n \log v^{max}(\M))$. 
\end{remark}

%% file: 6.appendix.tex
\appendix
\section{Missing Examples from Introduction}\label{app:intro}
The following example shows that given a fair division instance $\MMSinssym$, the $\MMS$ value of an agent $i \in \N$ can be as high as $v_i(\M)$.
\begin{example}
\label{eg:splc-mms-high}
Consider a fair division instance $\MMSinssym$ with $|\N| = n$. $\M$ consists of $n$ copies of same good, $|\M| = n$. The agents have a value of $1$ for any single copy and have value zero for any further copies she gets. In this case, $\MMS_i = 1$ for all agents by creating $n$ bundles with one copy of the good in each bundle. We also see that $v_i(\M) = 1$ for all $i \in [n]$.
\end{example}
The following example shows an instance with $\SPLC$ valuation functions where greedy (with the modification of stopping at $\frac{1}{2}$-$\MMS$) does not give $\frac{1}{2}$-$\MMS$ to all agents.
\begin{example}\label{eg:greedy-half-splc-counter}
    Consider an instance with $4$ agents and $8$ goods. There are $4$ copies of each good. We define valuation functions as follows.

    \noindent
    Agent $1$: $v_{111} = v_{121} = \frac{1}{2}-\delta$. $v_{131} = 2 \delta$. All other values are $0$. That is, agent $1$ likes one copy each of first three goods at the values defined and has zero value for all other copies and does not like any other goods.

    \noindent
    Agents $2$ and $3$: For $i \in \{2, 3\}$, $v_{i11} = 2\delta$. $v_{i21} = v_{i22} = v_{i31} = v_{i32} = \frac{1}{2} - \delta$. All other values are $0$. That is, agents $2$ and $3$ value only one copy of first good at value of $2\delta$ and first two copies of goods $2$ and $3$ at value of $\frac{1}{2}-\delta$. Any further copies are valued at $0$. The agents also do not value any other goods.

    \noindent
    Agent $4$: This agent has an additive value of $\frac{1}{8}$ for all goods, all copies.

    We see that the $\MMS$ value is $1$ for all agents (In fact, the value of concave extension at $(\frac{1}{4}), \ldots, \frac{1}{4}$ is $1$ for all agents.) Consider a greedy process that works as follows: Agent $1$ gets one copy of $g_1$ at $\frac{1}{2} - \delta$. Agent $2$ gets two copies of $g_2$ at $\frac{1}{2} - \delta$. Agent $2$ gets removed here. Agent $3$ gets two copies of $g_2$ at $\frac{1}{2} - \delta$. Agent $3$ gets removed at this point. Finally, agent $4$ is seeing a highest marginal value of $\frac18$ and agent $1$ is seeing highest marginal value of $2 \delta$. Therefore, if $\delta < \frac{1}{16}$, agent $4$ gets allocated goods in greedy and can pick all four copies of $g_3$. Therefore, agent $1$ does not get a value of $\frac{1}{2}$.
\end{example}
\section{Missing Details from Preliminaries}\label{app:prelims}
\subsection{Multilinear Extension}
Multilinear extension, $\mathbb{F}(\cdot)$ of a function $f(\cdot)$ is defined as follows.
\begin{definition}[Multilinear Extension of $f$]
    \begin{equation}
    \mathbb{F}(x) \coloneqq \sum_{S \subseteq V} \prod_{j \in S} x_j \prod_{j \in V \setminus S} (1 - x_j) f(S)
    \end{equation}
\end{definition}

Both Concave extension and Multilinear extension can be defined for any arbitrary set functions, they satisfy nice properties when functions are submodular.

Multilinear extension has been widely used in previous literature on fair division. However, we prove the following claim which shows that under Multilinear extension with submodular valuations, it is possible to have instances where no fractional allocation is $1$-$\MMS$.

\begin{claim}
There are instances where no fractional $1$-$\MMS$ allocation exists for submodular valuations under multilinear extension.
\end{claim}
\begin{proof}
    Consider a fair division instance $\MMSinssym$ where all valuations $v_i(\cdot)$ are monotone, non-negative submodular functions. Recall that $\MMS_i$ denotes the $\MMS$ value of agent $i$ in this instance. Consider new valuation functions, $\hat{v}_i(\cdot)$ defined as $\hat{v}_i(S) \coloneqq \min \{v_i(S), \MMS_i\}$. Note that $\hat{v}_i(\cdot)$ are also monotone, non-negative submodular functions and the $\MMS$ value of agent $i$ is retained.

    Suppose for contradiction, there exists a fractional allocation that under Multilinear extension gives $1$-$\MMS$ to each agent. Consider the Fair Division instance $(\N, \M, (\hat{v}_i)_{i \in \N})$. Let $\vecx = (x_{ij})_{i \in \N, j \in \M}$ denote the fractional allocation such that for all $i \in [n]$,  $\hat{v}_i(\vecx_i) = \MMS_i$. Note that $\vecx$ is also the social welfare maximizing allocation. Therefore we can round $\vecx$ using Pipage rounding \cite{ageev2004pipage} without any loss of value to any agent giving us an integral allocation that is $1$-$\MMS$. In particular, we show that if there is a fractional $1$-$\MMS$ allocation under Multilinear extension then there is an integral $1$-$\MMS$ allocation. However, it is well known that $1$-$\MMS$ allocation might not exist with submodular valuations. Therefore, the claim stands proved.
\end{proof}

\section{Missing Proofs From Section \ref{sec:splc} : Computational Aspects}\label{app:splc}
Recall from Section \ref{sec:alg-half-mms} we use $\mu_i$ values as follows:
\begin{equation}\label{eqn:app-mu-i}
    \mu_i = \sum_{j \in [t]} \left(\sum_{k \leq \frac{k_j}{n}}v_{ijk} + v_{ij{\lceil\sfrac{k_j}{n} \rceil}}\left(\frac{k_j}{n} - \left\lfloor \frac{k_j}{n} \right\rfloor \right)\right)
\end{equation}
Here we give full details of how to use the above defined $(\mu_i)$ values to get a polynomial time algorithm for $\frac{1}{2}$-$\MMS$ allocation. 
We prove the following two claims about these $(\mu_i)_{i \in \N}$.
\begin{lemma}\label{lem:mu_i_values}
    Given an $\SPLC$ instance, $\SPLCins$, $(\mu_i)_{i \in \N}$ defined in Equation \ref{eqn:app-mu-i} are \emph{feasible}.
\end{lemma}
\begin{proof}
    For any good $j \in [t]$, we give $\frac{k_j}{n}$ copies of the good to each agent. These clearly form a valid partition. Further, by definition, the value that the agent receives from $\frac{k_j}{n}$ copies of good $j$ is the value $\mu_i$. Therefore, these are feasible.
\end{proof}
\begin{lemma}
    Given an $\SPLC$ instance, $\SPLCins$, for all $i \in \N$, $\mu_i$ defined in Equation \ref{eqn:app-mu-i} satisfy $\mu_i \leq \APS_i \leq \MMS_i$.
\end{lemma}
\begin{proof}
    The value of $\mu_i$ defined in Equation \ref{eqn:app-mu-i} is the value of concave extension of $v_i(\cdot)$ at $(\frac{1}{n}, \ldots, \frac{1}{n})$. Therefore, by Lemma \ref{lem:aps-concave-closure}, we get that $\mu_i \leq \APS_i$. Combining this with Claim \ref{clm:mms-aps-rel}, we get the lemma.
\end{proof}
Algorithm \ref{alg:splc-mms-complete} describes our complete algorithm for computing $\frac{1}{2}$-$\MMS$ allocation.
\begin{algorithm}[tbh!]
\caption{Polynomial time $\frac{1}{2}$-$\MMS$ Algorithm for  $\SPLC$ valuations}\label{alg:splc-mms-complete}
\DontPrintSemicolon
  \SetKwFunction{Define}{Define}
  \SetKwInOut{Input}{Input}\SetKwInOut{Output}{Output}
  \Input{$\SPLCins$}
  \Output{Allocation $A$ where for every agent $v_i(A_i)\ge \MMS_i/2$}
  \BlankLine
  Compute values $\mu_i$ for all $i \in \N$ as given in Equation \ref{eqn:app-mu-i}\;
  Initialize set of \emph{active} agents $\A = \N$ and set of \emph{active} goods $\G = \M$.\;
  \While{There exists $i \in \N$ and $j \in [t]$ such that $v_{ij1} \geq \frac{\mu_i}{2}$} {\label{step:single-good-start-complete}
    $A_i \leftarrow \{j\}$\;
    $\A \leftarrow \A \setminus \{i\}$\;
    Remove one copy of $j$ from $\G$\;
    Re-compute $(\mu_i)_{i \in \A}$ using Equation $\ref{eqn:app-mu-i}$ for instance $(\A, \G, (v_i)_{i \in \A})$.
  }\label{step:single-good-end-complete}
  Solve the Linear Feasibility Program (\ref{lp:splc}) for the reduced instance $(\A, \G, (v_i)_{i \in \N})$ and obtain fractional allocation $\vecx$.\;
  Use the Rounding subroutine described in Algorithm \ref{alg:round} to get an integral allocation from $\vecx$.\;
  Return $A = (A_i)_{i \in \N}$.
\end{algorithm}
We can now state and prove the main theorem of this section.
\algsplcmmscomplete*
\begin{proof}
Algorithm \ref{alg:splc-mms-complete} is the polynomial time algorithm that achieves the stated guarantees. We prove this here.

\textbf{$\frac{1}{2}$-$\MMS$ guarantee.} At the start of while loop in steps \ref{step:single-good-start-complete} to \ref{step:single-good-end-complete}, the value of $\mu_i$ are at least as much as $\MMS_i$. Therefore, the very first agent gets a good of value at least $\frac{\MMS_i}{2}$. For the reduced instance, by Claim \ref{clm:single-good-redn}, the $\MMS$ values in the reduced instance can only increase. Therefore, the new $\mu_i$ that upper bound the $\MMS$ values in new instance are an upper bound on the $\MMS$ values of original instance. Thus, inductively, in every iteration, $\mu_i$ are a bound on $\MMS_i$ values. Therefore, all agents who get a good in the while loop of steps \ref{step:single-good-start-complete} to \ref{step:single-good-end-complete} get a value of at least $\frac{\MMS_i}{2}$. Further, the set of Active agents that remain at the end of this while loop also have $\mu_i$ values upper bounding the $\MMS_i$ values. Therefore, in a proof similar to proof of Theorem \ref{thm:half-mms-exist}, we get that every agent receives a bundle of value at least half their $\MMS$ value. 

\noindent
\textbf{Time Complexity.} Since $\mu_i$ form a closed form expression as defined in Equation \ref{eqn:app-mu-i}, they are polynomial time computable. All other steps are polynomial time as seen in proof of Theorem \ref{thm:half-mms-exist}. Therefore, the entire algorithm runs in polynomial time.
\end{proof}

%% file: main.bbl
\newcommand{\etalchar}[1]{$^{#1}$}
\begin{thebibliography}{GMSV12b}

\bibitem[AAB{\etalchar{+}}22]{AmanatidisABFLMVW22survey}
Georgios Amanatidis, Haris Aziz, Georgios Birmpas, Aris Filos-Ratsikas, Bo~Li, Herv{\'e} Moulin, Alexandros~A Voudouris, and Xiaowei Wu.
\newblock Fair division of indivisible goods: A survey.
\newblock {\em arXiv preprint arXiv:2208.08782}, 2022.

\bibitem[AD54]{AD}
K.~Arrow and G.~Debreu.
\newblock Existence of an equilibrium for a competitive economy.
\newblock {\em Econometrica}, 22:265--290, 1954.

\bibitem[AG23]{AkramiG23}
Hannaneh Akrami and Jugal Garg.
\newblock Breaking the $3/4$ barrier for approximate maximin share, 2023.

\bibitem[AGT23]{AkramiGS23}
Hannaneh Akrami, Jugal Garg, and Setareh Taki.
\newblock Improving approximation guarantees for maximin share, 2023.

\bibitem[AS04]{ageev2004pipage}
Alexander~A Ageev and Maxim~I Sviridenko.
\newblock Pipage rounding: A new method of constructing algorithms with proven performance guarantee.
\newblock {\em Journal of Combinatorial Optimization}, 8:307--328, 2004.

\bibitem[BEF21]{BabaioffEF21}
Moshe Babaioff, Tomer Ezra, and Uriel Feige.
\newblock Fair-share allocations for agents with arbitrary entitlements.
\newblock In {\em Proceedings of the 22nd ACM Conference on Economics and Computation}, pages 127--127, 2021.

\bibitem[BK20]{barman2020approximation}
Siddharth Barman and Sanath~Kumar Krishnamurthy.
\newblock Approximation algorithms for maximin fair division.
\newblock {\em ACM Transactions on Economics and Computation (TEAC)}, 8(1):1--28, 2020.

\bibitem[BKM17]{barman2017approximation}
Siddharth Barman and Sanath~Kumar Krishna~Murthy.
\newblock Approximation algorithms for maximin fair division.
\newblock In {\em Proceedings of the 2017 ACM Conference on Economics and Computation}, pages 647--664, 2017.

\bibitem[Bud11]{budish2011combinatorial}
Eric Budish.
\newblock The combinatorial assignment problem: Approximate competitive equilibrium from equal incomes.
\newblock {\em Journal of Political Economy}, 119(6):1061--1103, 2011.

\bibitem[BV20]{barman2020existence}
Siddharth Barman and Paritosh Verma.
\newblock Existence and computation of maximin fair allocations under matroid-rank valuations.
\newblock {\em arXiv preprint arXiv:2012.12710}, 2020.

\bibitem[CCG{\etalchar{+}}22]{SPLC4}
Bhaskar~Ray Chaudhury, Yun~Kuen Cheung, Jugal Garg, Naveen Garg, Martin Hoefer, and Kurt Mehlhorn.
\newblock Fair division of indivisible goods for a class of concave valuations.
\newblock {\em Journal of Artificial Intelligence Research}, 74:111--142, 2022.

\bibitem[FST21]{FeigeST21}
Uriel Feige, Ariel Sapir, and Laliv Tauber.
\newblock A tight negative example for mms fair allocations.
\newblock In {\em International Conference on Web and Internet Economics}, pages 355--372. Springer, 2021.

\bibitem[GHS{\etalchar{+}}18]{ghodsi2018fair}
Mohammad Ghodsi, MohammadTaghi HajiAghayi, Masoud Seddighin, Saeed Seddighin, and Hadi Yami.
\newblock Fair allocation of indivisible goods: Improvements and generalizations.
\newblock In {\em Proceedings of the 2018 ACM Conference on Economics and Computation}, pages 539--556, 2018.

\bibitem[GMSV12a]{SPLC2}
Jugal Garg, Ruta Mehta, Milind Sohoni, and Vijay~V Vazirani.
\newblock A complementary pivot algorithm for markets under separable, piecewise-linear concave utilities.
\newblock In {\em Proceedings of the forty-fourth annual ACM symposium on Theory of computing}, pages 1003--1016, 2012.

\bibitem[GMSV12b]{garg2012complementary}
Jugal Garg, Ruta Mehta, Milind Sohoni, and Vijay~V Vazirani.
\newblock A complementary pivot algorithm for markets under separable, piecewise-linear concave utilities.
\newblock In {\em Proceedings of the forty-fourth annual ACM symposium on Theory of computing}, pages 1003--1016, 2012.

\bibitem[GMT18]{garg2018approximating}
Jugal Garg, Peter McGlaughlin, and Setareh Taki.
\newblock Approximating maximin share allocations.
\newblock In {\em 2nd Symposium on Simplicity in Algorithms (SOSA 2019)}. Schloss Dagstuhl-Leibniz-Zentrum fuer Informatik, 2018.

\bibitem[GT20]{garg2019improved}
Jugal Garg and Setareh Taki.
\newblock An improved approximation algorithm for maximin shares.
\newblock In {\em Proceedings of the 21st ACM Conference on Economics and Computation}, EC ’20, page 379–380, 2020.

\bibitem[KKM23]{kulkarni2023maximin}
Pooja Kulkarni, Rucha Kulkarni, and Ruta Mehta.
\newblock Maximin share allocations for assignment valuations.
\newblock In {\em Proceedings of the 2023 International Conference on Autonomous Agents and Multiagent Systems}, pages 2875--2876, 2023.

\bibitem[KMT21]{KulkarniMT21}
Rucha Kulkarni, Ruta Mehta, and Setareh Taki.
\newblock Indivisible mixed manna: On the computability of mms+ po allocations.
\newblock In {\em Proceedings of the 22nd ACM Conference on Economics and Computation}, pages 683--684, 2021.

\bibitem[LV21]{LiV21}
Zhentao Li and Adrian Vetta.
\newblock The fair division of hereditary set systems.
\newblock {\em ACM Transactions on Economics and Computation (TEAC)}, 9(2):1--19, 2021.

\bibitem[MS04]{SPLC3}
Thomas~L Magnanti and Dan Stratila.
\newblock Separable concave optimization approximately equals piecewise linear optimization.
\newblock In {\em International Conference on Integer Programming and Combinatorial Optimization}, pages 234--243. Springer, 2004.

\bibitem[PW14]{ProcacciaW14}
Ariel~D Procaccia and Junxing Wang.
\newblock Fair enough: Guaranteeing approximate maximin shares.
\newblock In {\em Proceedings of the fifteenth ACM conference on Economics and computation}, pages 675--692. ACM, 2014.

\bibitem[UF23]{uziahu2023fair}
Gilad~Ben Uziahu and Uriel Feige.
\newblock On fair allocation of indivisible goods to submodular agents.
\newblock {\em arXiv preprint arXiv:2303.12444}, 2023.

\bibitem[VY11]{SPLC1}
Vijay~V Vazirani and Mihalis Yannakakis.
\newblock Market equilibrium under separable, piecewise-linear, concave utilities.
\newblock {\em Journal of the ACM (JACM)}, 58(3):1--25, 2011.

\bibitem[VZ22]{viswanathan2022yankee}
Vignesh Viswanathan and Yair Zick.
\newblock Yankee swap: a fast and simple fair allocation mechanism for matroid rank valuations.
\newblock {\em arXiv preprint arXiv:2206.08495}, 2022.

\end{thebibliography}
